\documentclass{cccg19}

\usepackage{amssymb,amsmath,verbatim,graphicx,tikz,hyperref}
\usetikzlibrary{fadings}

\newcommand{\conv}{\mathrm{conv}}

\newcommand{\vol}{\mathrm{vol}}

\newcommand\blfootnote[1]{%
  \begingroup
  \renewcommand\thefootnote{}\footnotetext{#1}%
  \addtocounter{footnote}{-1}%
  \endgroup
}

\title{Chirotopes of Random Points in Space are Realizable on a Small Integer Grid}
\author{Jean Cardinal\thanks{Universit\'e libre de Bruxelles (ULB), Brussels, Belgium. {\tt jcardin@ulb.ac.be}}
\and Ruy Fabila-Monroy\thanks{CINVESTAV, CDMX, Mexico. Partially supported by CONACyT grant 253261 {\tt ruyfabila@math.cinvestav.mx, cmhidalgo@math.cinvestav.mx} } 
\and Carlos Hidalgo-Toscano\footnotemark[2]}

\begin{document}
\maketitle
\blfootnote{\hspace{-0.5cm}\begin{minipage}[l]{0.1\textwidth} \includegraphics[trim=10cm 6cm 10cm 5cm,clip,scale=0.15]{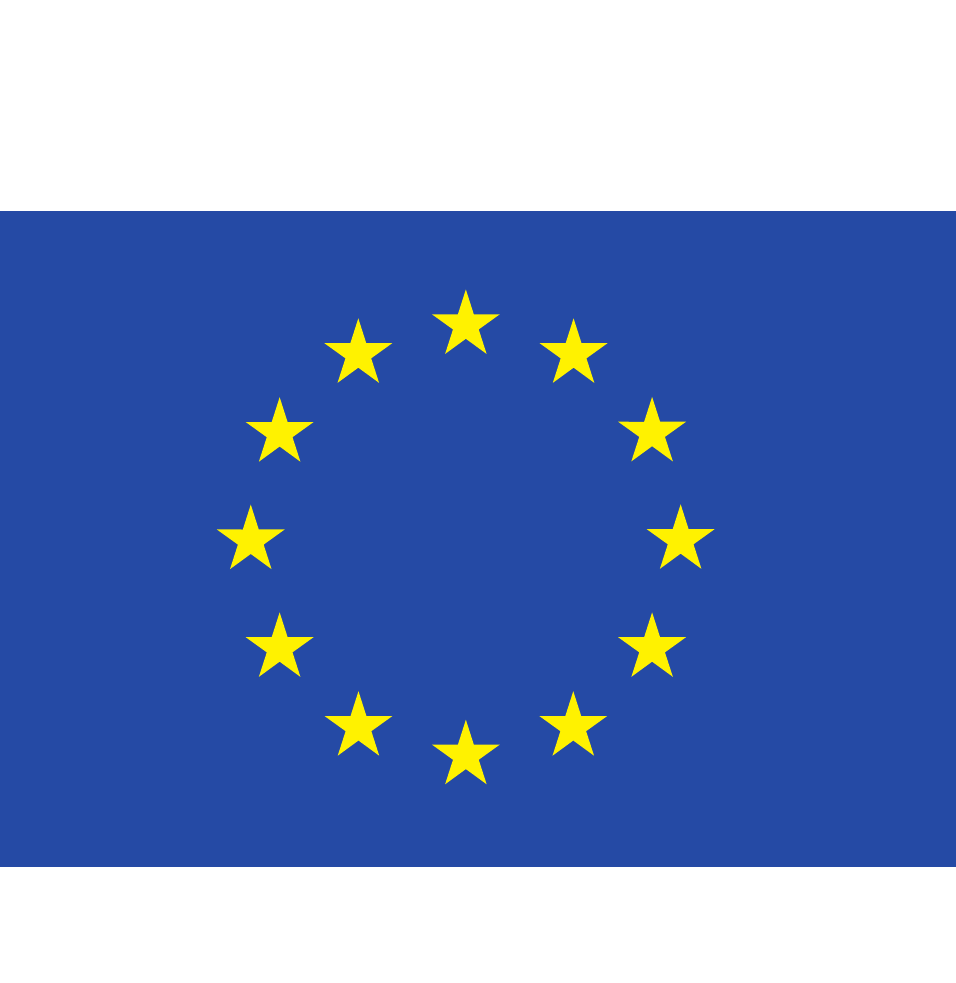} \end{minipage}  \hspace{-0.35cm} \begin{minipage}[l][1cm]{0.38\textwidth}
 	  This project has received funding from the European Union's Horizon 2020 research and innovation programme under the Marie Sk\l{}odowska-Curie grant agreement No 734922.
 	\end{minipage}}

\sloppy
\begin{abstract}
We prove that with high probability, a uniform sample of $n$ points in a convex domain in $\mathbb{R}^d$
can be rounded to points on a grid of step size proportional to $1/n^{d+1+\epsilon}$ without changing the underlying chirotope (oriented matroid). 
Therefore, chirotopes of random point sets can be encoded with $O(n\log n)$ bits.
This is in stark contrast to the worst case, where the grid may be forced to have step size 
$1/2^{2^{\Omega(n)}}$ even for $d=2$. 

This result is a high-dimensional generalization of previous results on order types of random planar point sets due
to Fabila-Monroy and Huemer (2017) and Devillers, Duchon, Glisse, and Goaoc (2018).
%
\end{abstract} 

\section{Introduction}

\paragraph{Chirotopes, order types, and oriented matroids.}

Many interesting properties of planar point sets in general position are captured by the combinatorial abstraction consisting of
the orientation -- clockwise or counterclockwise -- of every triple of points.
This generalizes naturally to $d$-dimensional point sets.
Consider a set $S\subset\mathbb{R}^d$ of $n$ points in general position (no $d+1$ on a hyperplane).
Let us denote by $\Lambda (S,k)$ the set of all ordered $k$-tuples of distinct points of $S$.
With every ordered $(d+1)$-tuple $P=(p_1,p_2,\ldots ,p_{d+1})\in \Lambda (S,d+1)$ of points we can associate a binary value $\chi (P)$ indicating 
the orientation of the corresponding simplex.
This can be expressed as the sign of a determinant:
\[
\chi (P)=
\mathrm{sgn} \left|
\begin{array}{ccccc}
p_{1,1} & p_{1,2} & \ldots & p_{1,d} & 1 \\
p_{2,1} & p_{2,2} & \ldots & p_{2,d} & 1 \\
& & \ddots & & 1 \\
p_{d+1,1} & p_{d+1,2} & \ldots & p_{d+1,d} & 1 \\
\end{array}
\right|.
\] 

The map $\chi$ is referred to as the {\em chirotope} of the point set $S$. 
The values of $\chi (P)$ obey the chirotope axioms, in particular the Grassmann-Pl\"ucker relations,
and completely characterize the rank-$(d+1)$ {\em oriented matroid} defined by the point set.
Note that not all maps satisfying the chirotope axioms are chirotopes of point sets. 
The {\em Topological Representation Theorem}, however, ensures that they always have a representation
as a collection of pseudohyperplanes. For $d=2$, such a representation is a {\em pseudoline arrangement}.
We refer the reader to the classical text from Bj{\"o}rner, 
Las Vergnas, Sturmfels, White, and Ziegler~\cite{BLSWZ99} for more background on oriented matroids and chirotopes.

We say that two sets of points have the same {\em order type} whenever there exists a bijection between
them such that the chirotope is preserved.
By extracting the purely combinatorial features of a set of points,  
oriented matroids and order types are useful tools in discrete and computational geometry.
In computational geometry, it is the case for instance that the chirotope of a point set is sufficient 
to compute its convex hull, and therefore provides a simple query-based computation model for this task, 
in a way that is reminiscent to comparison-based sorting. Knuth explored such a model and several
generalizations in his book {\em Axioms and Hulls}~\cite{K92}.
As for discrete geometry, Eppstein's recent book on forbidden configurations~\cite{E18} contains a thorough,
unified treatment of major results in geometry of planar point sets through the lens of monotone properties of chirotopes.
The {\em Order Type Database} from Aichholzer et al. contains all equivalence classes of chirotopes realized by sets of up to 10 points in the plane~\cite{AAK02}.

\paragraph{Algebraic universality and bit complexity.}

It is not clear, however, how much information is contained in a chirotope.
For $d=2$, it is known that a chirotope can be encoded using $O(n^2)$ bits~\cite{GP83,S97,F97,FV11}.
Recently, it has even been shown that chirotopes induced by sets of $n$ points in $\mathbb{R}^d$ could be stored using $O(n^d(\log \log n)^2/\log n)$ bits, in such a way that the orientation of every $(d+1)$-tuple can be recovered in $O(\log n/\log \log n)$ time on a word RAM~\cite{CCILO18}.

These representations, however, are distinct from the natural representation of a set of points by $d$-tuples of coordinates. 
The reason is that such a representation can have exponential bitsize: for every $n$ there exists
a chirotope of a set of $n$ points in the plane, every geometric realization of which requires coordinates that are doubly exponential in $n$~\cite{GPS89}. 
This is only one consequence of a more general phenomenon, known as {\em algebraic universality} of rank-three oriented matroids, and proved by Mn\"ev~\cite{M85,M88} and Richter-Gebert~\cite{RG95}.
In a nutshell, it states that for any semialgebraic set $A$, there exists an oriented matroid whose realization space 
is stably equivalent, in particular homotopic, to $A$. Algebraic universality holds for other discrete 
geometric structures such as unit disk graphs~\cite{MM13,KM12}, 4-dimensional polytopes~\cite{RG96}, 
simplicial polytopes~\cite{AP17}, and $d$-dimensional Delaunay triangulations~\cite{APT15}.

It is likely, however, that this worst-case exponential coordinate bitsize is irrelevant for ``typical'' point sets, 
or for point sets occurring in applications. It is therefore natural to wonder what is the required coordinate 
bitsize for chirotopes of {\em random} point sets.
The model we will assume here is that of a uniform distribution on a fixed compact, convex domain.

\paragraph{Related works.}

The question of random order types has first been tackled by Bokowski, Richter-Gebert, and Schindler~\cite{BRS91}.
They attribute the question of estimating the probability of an order type to Goodman and Pollack,
and investigate it under the assumption of a uniform distribution on the Grassmannian.
They discuss the problem of finding an efficient random generator on the Grassmann manifold from
a generator for the unit interval. They also perform experiments supporting the conjecture that the 
maximum probability is reached by the oriented matroid corresponding to the cyclic polytope. 

Recently, Fabila-Monroy and Huemer~\cite{FH17} proved that with high probability, a uniform sample of points in the plane can be rounded to
a $n^{3+\epsilon} \times n^{3+\epsilon}$ grid without altering its chirotope.
Even more recently, Devillers, Duchon, Glisse, and Goaoc~\cite{DDGG18} investigated the number of bits that need to be read from 
the coordinates of random points to know their order type, and obtain the same result in a slightly more general setting.
They also raise the question of whether uniform samples yield a vanishing fraction of order types. 
The difficulty of sampling order types uniformly was also discussed by Goaoc, Hubard, de Joannis de Verclos, Sereni, and Volec~\cite{GHVSV15}.

Another closely related line of work is that of random alignment and shape distribution of triangles~\cite{KK80,K85}. 
Probabilistic analyses supporting the idea that near-alignment of points occur naturally in random sets
were applied in particular to alignments of quasars~\cite{EG81}, and debunking pseudoscientific claims on
mysterious alignments between archaeological sites in Great Britain~\cite{WB83}.
It is known from these works, for instance, that for a uniform random sample of $n$ points in the unit square, 
the expected number of triples contained in a slab of width $w$ is proportional to $wn^3$. 
Hence unless the width is chosen to be at most proportional to $n^{-3}$, the sample contains near-aligned points, 
whose orientation is likely to be flipped by a rounding procedure.
%

\paragraph{Our contribution.}

We generalize the result of Fabila-Monroy and Huemer~\cite{FH17} and Devillers et al.~\cite{DDGG18} to $d$-dimensional point sets. 
We prove that in a uniform sample of $n$ points in $\mathbb{R}^d$, points can be rounded to a $n^{d+1+\epsilon}\times \ldots \times n^{d+1+\epsilon}$ grid
without altering their chirotope. 
We believe that the proof is simpler than the previous ones, even in the case $d=2$. 
%
%

\section{Chirotopes}

We begin with a simple observation on the structure of cells in an arrangement of $d+1$ hyperplanes in $\mathbb{R}^d$.

\begin{figure}
\begin{center}
\begin{tikzpicture}[scale=1.5]
    \fill[green, path fading=east] (1,0) -- (2,0) -- (2,-1) -- cycle; 
    \fill[green, path fading=north] (0,1) -- (0,2) -- (-1,2) -- cycle; 
    \fill[green, path fading=west, fading angle=45] (0,0) -- (-1,0) -- (0,-1) -- cycle; 
    \path[draw] (0,0) coordinate [label=above right:${\bf 0}$] 
            -- (1,0) coordinate [label=above right:$e_1$]
            -- (0,1) coordinate [label=above right:$e_2$]
            -- cycle;
\end{tikzpicture}

\begin{tikzpicture}[scale=1.5]
    \fill[green, path fading=west] (-1,2) -- (0,1) -- (0,0) -- (-1,0);
    \fill[green, path fading=south] (0,-1) -- (0,0) -- (1,0) -- (2,-1);
    \fill[green, path fading=east,fading angle=45] (0,2) -- (0,1) -- (1,0) -- (2,0);
    \path[draw] (0,0) coordinate [label=below left:${\bf 0}$]
            -- (1,0) coordinate [label=right:$e_1$]
            -- (0,1) coordinate [label=above:$e_2$]
            -- cycle;
\end{tikzpicture}
\end{center}
\caption{\label{fig:regions}The regions defined in Lemma~\ref{lem:obs} for $d=2$.
No line can intersect all the three $R_i$ or all the three $S_i$ regions simultaneously.}
\end{figure}
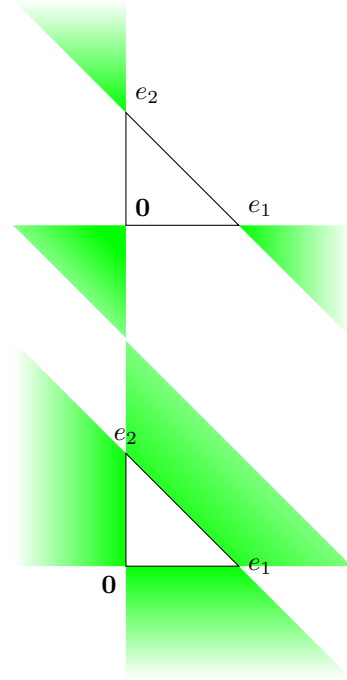

\begin{lemma} \label{lem:obs}
Let $P:=\{p_0,\dots,p_{d}\}$ be a set of $d+1$ points in general position in $\mathbb{R}^d$.
Let $\mathcal{H}$ be the hyperplane arrangement generated by all the hyperplanes 
passing through $d$ points of $P$. Let 
\begin{itemize}
\item $R_1,\dots,R_{d+1}$ be the unbounded
cells of $\mathcal{H}$ that do not contain a facet of $\conv(P)$ in their boundary; and

\item $S_1,\dots,S_{d+1}$ be the unbounded
cells of $\mathcal{H}$ that do  contain a facet of $\conv(P)$ in their boundary.
\end{itemize}
Then there is no hyperplane that simultaneously intersects all the $S_i$ or all the $R_i$.
\end{lemma}
\begin{proof}
 By doing an affine transformation we may assume that $p_0=0$ and $p_i$ is the vector
 with $1$ in its $i$-th coordinate and $0$ in all other coordinates. Note that the $R_i$
 and $S_i$ are defined by 
 \begin{eqnarray*}
R_i & := & \{x\in\mathbb{R}^d : x_j<0\ \forall j\in [d]\setminus\{i\}\wedge \sum_{j\in [d]}x_j>1 \},\\
R_{d+1} & := & \{x\in\mathbb{R}^d : x_j<0\ \forall j\in [d]\}
\end{eqnarray*}
and 
\begin{eqnarray*}
S_i & := & \{x\in\mathbb{R}^d : \\
& & x_i<0\wedge x_j>0\ \forall j\in [d]\setminus\{i\}\wedge \sum_{j\in [d]}x_j<1\},\\
S_{d+1} & := & \{x\in\mathbb{R}^d : x_j>0\ \forall j\in [d] \wedge \sum_{j\in [d]}x_j>1\}
\end{eqnarray*}
For $d=2$, the observation is direct and illustrated on Figure~\ref{fig:regions}.
For $d>2$, we proceed by induction and suppose that the result holds for $d-1$.
Consider a hyperplane $h$ intersecting the regions $R_1,R_2,\ldots ,R_d$. 
In both cases, if $h$ intersects $R_{d+1}$, then it must intersect $R_{d+1}\cap h'$, where $h'$ is one 
of the hyperplanes of equation $x_j=0$ for $j\in [d]$. 
The intersection of the whole arrangement with $h'$ yields a similar situation in dimension $d-1$, 
for which the statement holds by induction.
Therefore, $h$ cannot intersect $R_{d+1}$. The proof for the $S_i$ is similar.
\end{proof}

We now give a sufficient condition for two order tuples to have the same orientation after a perturbation.

\begin{lemma}
\label{lem:proj}
Consider two ordered $(d+1)$-tuples $P:=(p_1,p_2,\ldots ,p_{d+1})$ and $Q:=(q_1,q_2,\ldots ,q_{d+1})$ of points in $\mathbb{R}^d$. 
For every $1 \le i \le d+1$, let $f_i$ be the hyperplane passing through the facet of $\conv(P)$ opposite to $p_i$;

Suppose that for every $1 \le i \le d+1$ the following two conditions hold.
\begin{itemize}
 \item[1)] $p_i$ and $q_i$ are on the same open halfspace bounded by $f_i$; and
 \item[2)] the distances from $q_i$ to $f_i$ and from $p_i$ to $f_i$  are both larger than the distance
 from $q_j$ to $f_i$, for all $j \neq i$. 
\end{itemize}
Then $P$ and $Q$ have the same orientation.
\end{lemma}
\begin{proof}
By 1) and 2), for every $1 \le i \le d+1$ there exists a hyperplane $h_i$ parallel to $f_i$ that separates $p_i$ and $q_i$
from the other $q_j$ ($j \neq i$). Let $\mathcal{H}$ be the hyperplane arrangement generated by the $h_i$.
Note that for every $1 \le i \le d+1$, $p_i$ and $q_i$ lie in the same cell of $\mathcal{H}$.
Since the hyperplanes $h_i$ are parallel to the facets of $\conv(P)$, $\mathcal{H}$ can be of one of two types, depending on whether the unique bounded
cell of $\mathcal{H}$  has the same or the opposite orientation as $\conv(P)$, see Figure~\ref{fig:proj}.
Let $C_1,\dots,C_{d+1}$ be the cells of $\mathcal{H}$ containing $p_i$ and $q_i$, respectively. Note that
the $C_i$ are either the $R_i$ or the $S_i$ defined in  Lemma~\ref{lem:obs}. In both cases
no hyperplane can intersect all the $C_i$ simultaneously. For every $1 \le i \le d+1$, let  $f_i'$ be the hyperplane
passing through the $q_j$ different from $q_i$. Note, $f_i'$ intersects all the $C_j'$ with $j \neq i$. Thus,
$f_i'$ does not intersect $C_i$, and $p_i$ and $q_i$ are on the same open halfspace defined by $f_i$.
Therefore, $P$ and $Q$ have the same orientation.
\end{proof}

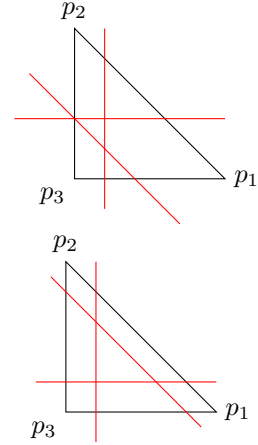
\begin{figure}
\begin{center}
\begin{tikzpicture}[scale=2]
    \path[draw] (0,0) coordinate [label=below left:$p_3$]
            -- (1,0) coordinate [label=right:$p_1$]
            -- (0,1) coordinate [label=above:$p_2$]
            -- cycle;
    \path[draw, color=red] (1-.3,0-.3) -- (0-.3,1-.3);
    \path[draw, color=red] (0-.4,0+.4) -- (1,0+.4);
    \path[draw, color=red] (0+.2,-.2) -- (0+.2,1);
\end{tikzpicture}

\begin{tikzpicture}[scale=2]
    \path[draw] (0,0) coordinate [label=below left:$p_3$]
            -- (1,0) coordinate [label=right:$p_1$]
            -- (0,1) coordinate [label=above:$p_2$]
            -- cycle;
    \path[draw, color=red] (1-.1,0-.1) -- (0-.1,1-.1);
    \path[draw, color=red] (0-.2,0+.2) -- (1,0+.2);
    \path[draw, color=red] (0+.2,-.2) -- (0+.2,1);
\end{tikzpicture}
\end{center}
\caption{\label{fig:proj}The hyperplanes $h_i$ in Lemma~\ref{lem:proj}.}
\end{figure}

We now prove our main result by showing that if $S$ is a random point set and $S'$ is obtained by rounding $S$
on a sufficiently dense grid, then the conditions of the lemma hold for every pair composed of a $d$-simplex in $S$
and its corresponding rounded version in $S'$.

\begin{theorem}\label{thm:main}
Let $S$ be a uniform sample of $n$ points in the $d$-dimensional unit ball.
Then for every $\epsilon > 0$, with probability at least $1-O\left(\frac{1}{n^{\epsilon}}\right )$, the points of $S$ can be rounded to a grid of step size
$1/(n^{d+1+\epsilon})$ without changing their chirotope.
\end{theorem}
\begin{proof}
Let $M:=n^{d+1+\epsilon}$.
Let $S'$ be the image of $S$ after rounding each point to its nearest neighbor on a grid of step size $1/M$.
Consider a $(d+1)$-tuple of points $P:=(p_1,\dots,p_{d+1})$ in $S$ and the corresponding 
$(d+1)$-tuple of rounded points $Q:=(q_1,\dots,q_{d+1})$ in $S'$. 
As in Lemma~\ref{lem:proj}, let $f_i$ be the hyperplane passing through the facet of $\conv(P)$ opposite to $p_i$.
We prove that the conditions of Lemma~\ref{lem:proj} hold with high probability. 
By definition, for any given $j$, the  absolute difference between $p_j$ and $q_j$ is at most $\sqrt{d}/M$.
Thus the conditions of Lemma~\ref{lem:proj} hold if the distance from $p_i$ to $f_i$ is at least $2\sqrt{d}/M$.
Let $B_d$ be the unit $d$-dimensional ball.
The $d-1$-volume of the intersection of $B_d$ and the hyperplane containing $f_i$ is at most $\vol(B_d-1)$.
Thus, the probability that for a given $1 \le i \le d+1$ the distance from $p_i$ to $f_i$ is less or equal to $2\sqrt{d}/M$ is at most  $(2\sqrt{d}/M)\vol (B_{d-1})/\vol (B_d)$. 
We have that
\begin{eqnarray*}
\frac{2\sqrt{d}}{M} \cdot \frac{\vol (B_{d-1})}{\vol (B_d)} & = & \frac{2\sqrt{d}}{\sqrt{\pi}M}\cdot \frac{\Gamma (\frac{d}{2}+1)}{\Gamma (\frac{d-1}{2}+1)} \\
& < & \frac{2\sqrt{d}}{\sqrt{\pi}M}\cdot \frac{\Gamma (\frac{d}{2}+1)}{\Gamma (\frac{d}{2})} \\
& = & \frac{d^{3/2}}{\sqrt{\pi}M}.
\end{eqnarray*}
We  apply the union bound over all such bad events. There are $(d+1)\binom{n}{d+1}$ such events to consider. 
Thus, the probability that no $(d+1)$-tuple has a different orientation as the corresponding $(d+1)$-tuple
in $S'$ is at least \[1-\binom{n}{d+1} \frac{(d+1)d^{3/2}}{\sqrt{\pi}M}=1-O\left(\frac{1}{n^{\epsilon}}\right ).\]
\end{proof}
Note that we considered the uniform distribution on the unit ball for convenience.
The same analysis holds for any fixed convex body, where the probability of a bad event happening depends on the
discrepancies in the distributions of the projections in different directions.

\bibliographystyle{plain}
\bibliography{resolutionRandom}

\end{document}